\documentclass{article}
\usepackage{amssymb, amsmath, amsthm,pstricks}
\input{defns.sty}

\newtheorem{default}{}

\newtheorem{corollary}[default]{Corollary}
\newtheorem{definition}[default]{Definition}
\newtheorem{example}[default]{Example}

\newtheorem{exer-hard}[default]{*Exercise}
\newtheorem{lemma}[default]{Lemma}

\newtheorem{proposition}[default]{Proposition}
\newtheorem{remark}[default]{Remark}
\newtheorem{theorem}[default]{Theorem}

\def\Notation{{\medskip\noindent{\bf Notation~}}}

\def\qed{{\hfill  $\Box$}}

\newboolean{SHORT}
\setboolean{SHORT}{false}

\newcommand\SHORTVERSION[2]{
	\ifthenelse{\boolean{SHORT}}{#1}{#2}
}

\newenvironment{DEF}{\begin{definition}\bgroup\rm}{\egroup\end{definition}}

\newenvironment{cor}{\begin{corollary}}{\end{corollary}}
\newenvironment{lem}{\begin{lemma}}{\end{lemma}}
\newenvironment{Thm}{\begin{theorem}}{\end{theorem}}

\def\true{1}

\def\NN{\mathbb N}   

\def\partialto{\rightharpoonup}
\def\oraclefun#1#2{\beta_{{#1},{#2}}}

\def\dom#1{\mathrm{dom}({#1})}
\def\rng#1{\mathrm{rng}({#1})}
\def\truefalse{\{0,1\}}

\def\Afun#1{\alpha_{#1}}
\def\Osection#1#2{{#2}^{[{#1}]}}
\def\Itlang#1#2{\mathcal L^{#1}_{#2}}

\begin{document}

\title{Relativizing Small Complexity Classes and their Theories\footnote{A preliminary version of this paper appeared as \cite{ACN:csl}.}}
\author{Klaus Aehlig\thanks{supported by DFG grant Ae 102/1-1}
        \and Stephen Cook\thanks{supported by NSERC}
\and Phuong Nguyen\thanks{supported by NSERC}}
\date{\today}

\maketitle
\thispagestyle{empty}

\begin{abstract}

Existing definitions of the relativizations of \NCOne, \L\ and \NL\
do not preserve the inclusions
$\NCOne \subseteq \L$, $\NL\subseteq \ACOne$.
We start by giving the first definitions that preserve them.
Here for \L\ and \NL\ we define their relativizations using Wilson's stack oracle model,
but limit the height of the stack to a constant (instead of $\log(n)$).
We show that the collapse of any two classes in
$\{\ACZm, \TCZ, \NCOne, \L, \NL\}$
implies the collapse of their relativizations.
Next we exhibit an oracle $\alpha$ that makes $\ACk(\alpha)$
a proper hierarchy.  This strengthens and clarifies the separations
of the relativized theories in [Takeuti, 1995].
The idea is that a circuit whose nested depth of oracle gates is bounded by $k$
cannot compute correctly the $(k+1)$ compositions of every oracle function.
Finally we develop theories that characterize the relativizations
of subclasses of \Ptime\ by modifying theories previously defined by
the second two authors.  A function is provably total in a theory iff
it is in the corresponding relativized class, and hence the oracle
separations imply separations for the relativized theories.

\end{abstract}

\section{Introduction}
\label{s:intro}

Oracles that separate \Ptime\ from \NP\ and oracles that
collapse \NP\ to \Ptime\ have both been constructed.
This rules out the possibility of proofs of the separation or collapse of 
\Ptime\ and \NP\ by methods that relativize.
However, similar results have not been established for subclasses of \Ptime\
such as \L\ and \NL.
Indeed, prior to this work there has not been a satisfying definition of
the relativized version of \NL\ that preserves simultaneously the
inclusions.

\begin{equation}\label{e:inclusions}
\NCOne \subseteq \L  \subseteq \NL \subseteq \ACOne
\end{equation}
(In this paper \NCk\ and \ACk\ refer to their uniform versions.)
For example~\cite{Ladner:Lynch:76} if the Turing machines are allowed to be nondeterministic when
writing oracle queries, then there is an oracle $\alpha$ so that
$\NL(\alpha) \not\subseteq \Ptime(\alpha)$.
Later definitions of $\NL(\alpha)$ adopt the requirement specified in \cite{Ruzzo:Simon:Tompa}
that the nondeterministic oracle machines be deterministic
whenever the oracle tape (or oracle stack) is nonempty.
Then the inclusion $\NL(\alpha) \subseteq \Ptime(\alpha)$ relativizes,
but not all inclusions in \eqref{e:inclusions}.

Because the nesting depth of oracle gates in an oracle \NCOne\ circuit
can be bigger than one,
the model of relativization that preserves the inclusion $\NCOne \subseteq \L$
must allow an oracle logspace Turing machine to have access to more than one
oracle query tape \cite{Orponen:83,Buss:86:sctc,Wilson:88:jcss}.
For the model defined by Wilson~\cite{Wilson:88:jcss}, the partially constructed
oracle queries are stored in a stack.
The machine can write queries only on the oracle tape at the top of the stack.
It can start a new query on an empty oracle tape (thus {\em pushing} down the
current oracle tape, if there is any),
or query the content of the top tape which then becomes empty and the stack is {\em popped}.

Following Cook~\cite{Cook:85:information-control}, the circuits accepting languages
in relativized \NCOne\ are those with logarithmic depth where the Boolean gates
have bounded fanin and an oracle gate of $m$ inputs contributes $\log(m)$ to
the depths of its parents.
Then in order to relativize the inclusion $\NCOne \subseteq \L$,
the oracle logspace machines defined by Wilson~\cite{Wilson:88:jcss}
are required to satisfy the condition that at any time,
$$\sum_{i=1}^k\max\{\log(|q_i|),1\} = \bigO(\log(n))$$
where $q_1, q_2, \ldots, q_k$ are the contents of the stack and $|q_i|$ are their lengths.
For the simulation of an oracle \NCOne\ circuit by such an oracle logspace machine
the upper bound $\bigO(\log(n))$ cannot be improved. 

Although the above definition of $\L(\alpha)$ (and $\NL(\alpha)$) ensures that $\NCOne(\alpha) \subseteq \L(\alpha)$,
unfortunately we know only that $\NL(\alpha) \subseteq \AC^2(\alpha)$ \cite{Wilson:88:jcss};
the inclusion $\NL(\alpha) \subseteq \ACOne(\alpha)$ is left open.

We observe that if the height of the oracle stack is bounded by a constant
(while the lengths of the queries are still bounded by a polynomial in the length of the inputs),
then an oracle \NL\ machine can be simulated by an oracle \ACOne\ circuit,
i.e., $\NL(\alpha) \subseteq \ACOne(\alpha)$.
In fact, it can then be shown that $\NL(\alpha)$ is contained in the $\ACZ(\alpha)$ closure of the Reachability problem
for directed graphs,
while $\L(\alpha)$ equals the $\ACZ(\alpha)$ closure of the Reachability problem
for directed graphs whose outdegree is at most one.

The \ACZalpha\ closure of the Boolean Sentence Value problem (which is \ACZ\ complete for \NCOne)
turns out to be the languages computable by uniform oracle \NCOne\ circuits (defined as before)
where the nesting depth of oracle gates is now bounded by a constant.
We redefine $\NCOne(\alpha)$ using this new restriction on the oracle gates;
the new definition is more suitable in the context of $\ACZ(\alpha)$ reducibility
(the previous definition of \NCOnealpha\ seems suitable when one considers $\NCOne(\alpha)$ reducibility).
Consequently, we obtain the first definition of $\NCOne(\alpha)$, $\L(\alpha)$ and $\NL(\alpha)$
that preserves the inclusions in \eqref{e:inclusions}.

Furthermore, the \ACZ-complete problems for \NCOne, \L, and \NL\
(as well as \ACZm, \TCZ) become $\ACZalpha$-complete for the corresponding relativized classes.
Therefore the existence of any oracle that separates two of the mentioned classes
implies the separation of the respective nonrelativized classes.
(If the non-relativised classes would be equal, their complete
problems would be equivalent under \ACZ-reductions, hence even
more under $\ACZalpha$-reductions and therefore the relativised
classes would coincide as well.)
Separating the relativized classes
is as hard as separating their nonrelativized counterparts.
This nicely generalizes known results~\cite{Wilson:88:jcss,Simon:77:thesis,Wilson:87:mst}. 

On the other hand, oracles that separate the classes \ACk\ (for $k = 0,1,2,\ldots$) and \Ptime\
have been constructed \cite{Wilson:87:mst}.
Here we prove a sharp separation between relativized circuit classes
whose nesting depths
of oracle gates differ by one.
More precisely, we show that a family of uniform circuits with nesting depth of oracle gates bounded by $k$
cannot compute correctly the $(k+1)$ iterated compositions
\begin{equation}\label{e:i-composition}
f(f(\,\ldots\,f(0)\,\ldots\,))
\end{equation}
for all oracle function $f$.
(Clearly \eqref{e:i-composition} can be computed correctly by a circuit with oracle gates
having nesting depth $(k+1)$.)
As a result, there is an oracle $\alpha$ such that
\begin{equation}\label{e:separation}
\NL(\alpha)\subsetneq \ACOne(\alpha) \subsetneq \AC^2(\alpha) \subsetneq \ldots \subsetneq \Ptime(\alpha)
\end{equation}

The idea of using \eqref{e:i-composition} to separate relativized circuit classes
is already present in the work of Takeuti~\cite{Takeuti:95:apal} where it is used to separate
the relativized versions of first-order theories $TLS(\alpha)$ and $TAC^1(\alpha)$.
Here $TLS$ and $TAC$ are (single sorted) theories associated with \L\ and \ACOne, respectively.
Thus with simplified arguments we strengthen his results.

Finally, building up from the work of the second two authors \cite{Cook:Nguyen,Nguyen:Cook:05:lmcs}
we develop relativized two-sorted theories that are associated with the newly defined classes
$\NCOnealpha, \L(\alpha), \NL(\alpha)$ as well as other relativized circuit classes.

The paper is organized as follows.
In Section \ref{s:class} we define the relativized classes
and prove the inclusions mentioned above.
An oracle that separates classes in \eqref{e:separation} is shown in Section \ref{s:sep}.
In Section \ref{s:theory} we define the associated theories and show
their separation using the oracle defined in Section \ref{s:sep}.

\section{Small Relativized Classes}
\label{s:class}

\subsection{Relativized Circuit Classes} \label{s:cir-class}

Throughout this paper, $\alpha$ denotes a unary relation on binary strings.

A problem is in \ACk\ if it can be solved by a polynomial size family
of Boolean circuits whose depth is
bounded by $\bigO((\log{n})^k)$ ($n$ is the number of input bits),
where $\wedge$ and $\vee$ gates are allowed unbounded fanin.
The relativized class $\ACkalpha$ generalizes this by allowing,
in addition to (unbounded fanin) Boolean gates ($\neg, \wedge, \vee$),
oracle gates
that output 1 if and only if the inputs to the gates (viewed as binary strings)
belong to $\alpha$ (these gates are also called $\alpha$ gates).

In this paper we always require circuit families to be
{\em uniform}.  Our default definition of uniform is DLOGTIME, a robust
notion of uniformity that has a number of equivalent definitions
\cite{Barrington:Immerman:Straubing:90,Immerman:99:book}. 
In particular, a language $L\subseteq \{0,1\}^*$
is in (uniform) \ACZ\ iff it represents the set
of finite models $\{1,\ldots,n\}$ of some fixed first-order
formula with an uninterpreted
unary predicate symbol and ternary predicates which are interpreted
as addition and multiplication.

Here an \ACZ\ reduction refers to a `Turing' style reduction.
Thus a problem $A$ is \ACZ\ reducible to a problem $B$ if there
is a uniform polynomial size constant depth family of circuits
computing $A$, where the circuits are allowed to have oracle gates
for $B$, as well as Boolean gates.

Recall that \TCZ\ (resp. $\ACZ(m)$) is defined in the same way as \ACZ, except
the circuits allow unbounded fanin threshold (resp. $\bmod m$) gates.

\begin{definition}[\ACkalpha, $\ACZ(m,\alpha)$, \TCZalpha]
\label{d:small_classes}
For $k \ge 0$, the class \ACkalpha\ (resp. $\ACZ(m,\alpha)$, \TCZalpha) is defined
as uniform \ACk\ (resp. $\ACZ(m)$, \TCZ) except that
unbounded fan-in $\alpha$ gates are allowed.
\end{definition}

The class \NCk\ is the subclass of \ACk\ defined by restricting
the $\wedge$ and $\vee$ gates to have fanin $2$.
Defining \NCkalpha\ is more complicated.
In \cite{Cook:85:information-control} the depth of an oracle gate
with $m$ inputs is defined to be $\log(m)$.
A circuit is an \NCkalpha-circuit provided that
it has polynomial size and the total depth of all gates 
along any path from the output gate to an input gate is $\bigO((\log{n})^k)$.
Note that if there is a mix of large and small oracle gates,
the nested depth of oracle gates may not be $\bigO((\log{n})^{k-1})$.

Here we restrict the definition further, requiring that the nested depth of
oracle gates is $\bigO((\log{n})^{k-1})$.
This restriction allows us to show that in the relativized world, \NCOne\ is still contained in \L,
and that the circuit value problem (for oracle \NCk\ circuits)
is still complete for \NCk\ as expected.

\begin{definition}[\NCkalpha]
\label{d:NCkalpha}
For $k \ge 1$, a language is in $\NCkalpha$ if it is computable by a uniform family of \NCkalpha\ circuits,
i.e., \ACkalpha\ circuits
where the $\wedge$ and $\vee$ gates have fanin 2, and the nested depth of $\alpha$ gates is $\bigO((\log{n})^{k-1})$.
\end{definition}

The following inclusions extend the inclusions of the nonrelativized classes:
\begin{equation*}
\ACZ(\alpha) \subseteq \ACZ(m,\alpha)
\subseteq \TCZ(\alpha) \subseteq \NCOne(\alpha) \subseteq \ACOne(\alpha)
\subseteq \ldots \subseteq \Ptime(\alpha)
\end{equation*}

\ignore{
The next lemma is straightforward:
\begin{lemma}\label{t:NC-alternative}
For $k \ge 1$, a language is in \NCkalpha\ if and only if it is computed by
a uniform family of circuits $\{C_n\}$, where each $C_n$
is constructed from a constant number of $\AC^{k-1}(\alpha)$-circuits
and \NCk-circuits $C^n_1, \ldots, C^n_t$ in the way that the outputs of $C^n_i$ are inputs to $C^n_{i+1}$ ($1\le i < t$),
the inputs of $C^n_1$ are the inputs of $C_n$, and the outputs of $C^n_t$ are the outputs of $C_n$.
\end{lemma}
}

Further the \ACZ-complete problems for \ACZ(m), \TCZ, and \NCOne\
are also \ACZalpha-complete for the corresponding relativized classes.
This is expressed by the next result, using the following complete
problems: \modm\ and \THRESH\ (the threshold function)  are \ACZ-complete
for \ACZ(m) and \TCZ\ respectively, and \FormVal\
(the Boolean formula value problem) is both \ACZ-complete and
\ACZ-many-one complete for \NCOne.

\begin{proposition}\label{p:complete}
\begin{eqnarray}
\ACZ(m,\alpha) & = & \ACZ(\modm,\alpha)\\
\TCZ(\alpha) & = & \ACZ(\THRESH,\alpha)\\
\NCOne(\alpha) & = & \ACZ(\FormVal,\alpha)
\end{eqnarray}
\end{proposition}

\begin{proof}
Each class on the right is included in the  corresponding class on
the left because
a query to the complete problem can be replaced by a circuit 
computing the query.  Each class on the left is a subset of the
corresponding class on the right
because the queries to $\alpha$ on the left have bounded
nesting depth.  
\end{proof}

Note that there is no similar characterization of $\ACOne(\alpha)$
or $\Ptime(\alpha)$,
because here the queries to $\alpha$ can have unbounded nesting depth.

\subsection{Relativized Logspace Classes}

To define oracle logspace classes, we use a modification of
Wilson's stack model \cite{Wilson:88:jcss}.
An advantage is that the relativized classes defined here are
closed under \ACZ-reductions.
This is not true for the non-stack model.

A Turing machine \sM\ with a stack of oracle tapes can write 0 or 1
onto the top
oracle tape if it already contains some symbols, or it can start writing on
an empty oracle tape.
In the latter case, the new oracle tape will be at the top of the stack,
and we say that \sM\ performs a {\it push} operation.
The machine can also {\it pop} the stack, and its next action and state depend
on $\alpha(Q)$, where $Q$ is the content of the top oracle tape.
Note that here the oracle tapes are write-only.

Instead of allowing an arbitrary number of oracle tapes,
we modify Wilson's model by allowing only a stack of constant height
(hence the prefix ``$\mbox{cs}$'' in \csLalpha\ and \csNLalpha).
This places the relativized classes in the same order as the order of
their unrelativized counterparts.

In the definition of \csNLalpha, we also use the restriction
\cite{Ruzzo:Simon:Tompa} that the machine
is deterministic when the oracle stack is non empty or when it
is in a push state.

\begin{definition}[\csLalpha, \csNLalpha]
\label{d:csL-NL}
For a unary relation $\alpha$ on strings, \csLalpha\ is the class
of languages computable by logspace, polytime Turing machines
using an $\alpha$-oracle stack whose height is bounded by a constant.
\csNLalpha\ is defined as \csLalpha\ but the Turing machines are allowed
to be nondeterministic when the oracle stack is empty.
\end{definition}

\begin{theorem}\label{t:order-a}
$\NCOne(\alpha) \subseteq \csLalpha \subseteq \csNLalpha \subseteq \AC^1(\alpha)$.
\end{theorem}

\begin{proof}
The second inclusion is immediate from the definitions,
the first can be proved as in the standard proof of the fact that $\NCOne \subseteq \L$
(see also \cite{Wilson:88:jcss}).
The last inclusion can actually be strengthened, as shown in the next theorem.
\end{proof}

The next theorem partly extends Proposition \ref{p:complete} to the two
new classes.  Recall that \STCONN\ is the problem: given $(G, s, t)$,
where $s, t$ are two designated vertices of a directed graph $G$,
decide whether there is a path from $s$ to $t$.  We define
\OneSTCONN\ to be the same, except we require that every node in $G$
has out degree at most one.   Then \STCONN\ and \OneSTCONN\
are \ACZ-many-one complete for \NL\ and \L, respectively.

A \csLalpha\ function is defined by allowing the \csLalpha\ machine to
write on a write-only output tape.
Then the notion of many-one \csLalpha\ reducibility is defined as usual.

\begin{theorem}\label{t:order-b}
\begin{description}
\item[(i)] $\csLalpha = \AC^0(\OneSTCONN,\alpha)$ 
\item[(ii)] $\csNLalpha \subseteq \AC^0(\STCONN,\alpha)$
\item[(iii)] A language is in \csNLalpha\ iff it is many-one
\csLalpha-reducible to \STCONN.
\end{description}
\end{theorem}

\begin{proof}

We start by proving the inclusion
$\AC^0(\OneSTCONN,\alpha) \subseteq \csLalpha$ in (i).
A problem in $\AC^0(\OneSTCONN,\alpha)$ is given by a uniform polynomial
size constant depth circuit family $\{C_n\}_{n\in\NN}$ with oracle
queries to \OneSTCONN\ and $\alpha$.  A \csLalpha\ machine \sM\
on an input $x$ of length $n$ performs a depth-first search of the
circuit $C_n$ with input $x$.  Each $\alpha$ oracle gate at depth $k$
is answered using an oracle query at stack height $k$, and each
oracle query to \OneSTCONN\ is answered by a log-space computation.

Note that this argument does not work for the corresponding
inclusion in (ii), because once the oracle stack is nonempty a
\csNLalpha\ machine becomes deterministic and cannot answer
oracle queries to \STCONN\ (assuming $\L \ne \NL$).

Now we prove the inclusion (ii).  (The corresponding inclusion in
the equation (i) is proved similarly.)
Let \sM\ be a nondeterministic logspace Turing machine with a constant-height stack of oracle tapes.
Let $h$ be the bound on the height of the oracle stack.
There is a polynomial $p(n)$ so that for each input length $n$
and oracle $\alpha$, \sM\ has at most $p(n)$ possible configurations:
\begin{equation}\label{e:csNL-vertices}
u_0 = \mathit{START},\ u_1 = \mathit{ACCEPT},\ u_2, \ldots, u_{p(n)-1}
\end{equation}
Here a configuration $u_i$ encodes information about the internal
state, the
content and head position of the work tape, and the position of the input
tape head, {\em but no explicit information about the oracle stack}
(although the internal state might encode implicit information).

Given an input $x$ of length $n$ we construct a sequence $G_0,\cdots,G_h$
of directed graphs such that the set $V_k$ of nodes in $G_k$ consists of
all pairs $(k,u)$, where $u$ is a configuration and $k$ is the current
height of the oracle stack (so $0\le k\le h$).  Thus $(k,u)$
represents a `height $k$ configuration'.  Note that the
computation of \sM\ on input $x$ can be described by a sequence of
nodes in $\bigcup_k V_k$.  We want to define the edge set $E_k$ so
that
\begin{equation}\label{e:edgeDef}
\mbox{$((k,u),(k,u'))\in E_k$ iff $(k,u')$ can be the next height $k$
configuration after $(k,u)$}
\end{equation}

The edges $E_k$ in $G_k$ comprise the union
$$
     E_k = E^0_k \cup E^1_k
$$
Define $E^0_k$ to consist of all pairs $((k,u),(k,u'))$ such that
$u$ does not cause a push or pop and $u'$ is a possible successor to $u$. 
If $k\ge 1$ then we also require that $u$ be a deterministic
configuration.

Define $E^1_k$ to be empty if $k=h$, and if $0\le k < h$ then $E^1_k$
consists of all pairs $((k,u),(k,u'))$ such that $u$ is a deterministic
configuration causing a push, and there is a sequence
\begin{equation}\label{e:vSeq}
    (k+1, v_0),\ldots,(k+1,v_t)
\end{equation}
of configurations such that $v_0$ is the successor
of $u$ and $((k+1,v_i), (k+1,v_{i+1})) \in E_{k+1}$ for $0\le i<t$
and $v_t$ causes a pop and $u'$ is the successor of $v_t$ (given
that $(k,u)$ is the most recent level $k$ node preceding $(k+1,v_t)$).

Note that in a computation from $(k,u)$ to $(k,u')$ the sequence
(\ref{e:vSeq}) is the sequence of height $k+1$ nodes, and this
sequence determines the string $Q$ that is written on the the height $k+1$
oracle tape, and hence determines whether $u'$ is the successor of $v_t$.

It is easy to prove (\ref{e:edgeDef}) by induction on
$k=h,h-1,\ldots,0$.  For $k=0$ this implies that $\sM$ accepts $x$
iff there is a path in $G_0$ from $(0,\mathit{START})$ to
$(0,\mathit{ACCEPT})$.  Thus it suffices to show that some
$AC^0(\STCONN,\alpha)$ circuit computes the adjacency matrix
of each graph $G_k$ given the input $x$.   In fact it is easy to
see that some $AC^0$ circuit with input $x$ outputs the
adjacency matrix for each edge set $E^0_k$.  Hence it suffices
to show that some $AC^0(\STCONN,\alpha)$ circuit computes the
adjacency matrix for $E^1_k$ given input $x$ and the adjacency
matrix for $E^1_{k+1}$, $0\le k < h$.  This can be done since
the the elements for the sequence (\ref{e:vSeq}) can be obtained
from $E_{k+1}$ using oracle queries to $\STCONN$, and the
string $Q$ that is written on the the height $k+1$ oracle tape
can be extracted from this sequence using an $AC^0$ circuit,
and so $\alpha(Q)$ can be used to determine $u'$ is the successor
of $v_t$.

Note that the depth of nesting of oracle calls to $\alpha$ is $h$.

To prove (iii) we note that the direction $(\Leftarrow)$
is easy:  A \csNLalpha\ machine \sM\ on input $x$ answers the single
query $f(x,\alpha)$ to \STCONN\
by simulating the \NL\ machine $\sM'$ that answers the query on
input $f(x,\alpha)$, and
each time $\sM'$ requires another input bit, \sM\ deterministically
computes that bit.

To prove (iii) in the direction $(\Rightarrow)$ we note that the
edge relation $E_0$ defined in the proof of (ii) can be computed
by an \csLalpha\ machine.  Then as noted above, $\sM$ accepts its input
$x$ iff there is a path in $G_0$ from $(0,\mathit{START})$ to $(0,\ACCEPT)$,
which is an instance of \STCONN.
\end{proof}

\begin{corollary}\label{c:separ}
The existence of an oracle $\alpha$ separating any two of the classes
$$
\ACZ(m,\alpha)\subseteq \TCZ(\alpha)\subseteq \NCOne(\alpha)\subseteq 
    \csLalpha \subseteq \csNLalpha 
$$
implies the separation of the respective nonrelativized classes.
\end{corollary}
\begin{proof}
This follows from Proposition \ref{p:complete} and Theorem \ref{t:order-b}
parts (i) and (ii).  If any two of the nonrelativized classes is equal,
then the corresponding complete problems would be \ACZ-equivalent,
and hence the relativized classes would be equal.
\end{proof}

\begin{corollary}[Relativized Immerman-\Sze\ Theorem]
\csNLalpha\ is closed under complementation.
\end{corollary}

\begin{proof}
By Theorem \ref{t:order-b} (iii)
any language in \co{\csNLalpha} is many-one \csLalpha\ reducible to
$\overline{\STCONN}$, which is many-one \ACZ\ reducible to \STCONN.
So $\co{\csNLalpha} \subseteq \csNLalpha$.
\SHORTVERSION{\qed}{}
\end{proof}

Let $\csL^2(\alpha)$ denote the class of languages computable by
a deterministic oracle Turing machine in $\bigO(\log^2)$ space and
constant-height oracle stack.

\begin{corollary}[Relativized Savitch's Theorem]
$\csNLalpha \subseteq \csL^2(\alpha)$.
\end{corollary}

\begin{proof}
The corollary follows from Theorem \ref{t:order-b} (iii)
and the fact that the composition of a \csLalpha\ function and a
$(\log^2)$ space function (for \STCONN) is a $\csL^2(\alpha)$ function.
\end{proof}

It is easy to see that the function class associated with either of
the classes $\AC^0(\OneSTCONN,\alpha)$ or $\AC^0(\STCONN,\alpha)$
is closed under composition, so we have the following result.

\begin{corollary}\label{c:cslComp}
The function class associated with $\csL(\alpha)$ is closed under
composition.
\end{corollary}

However it is an open question whether the function class associated
with $\csNLalpha$ is necessarily closed under composition, for
the same reason that we cannot necessarily conclude that inclusion
in part (ii) of Theorem \ref{t:order-b} can be changed to equality.
Once the oracle stack in a $\csNLalpha$ machine becomes nonempty
it becomes deterministic, so it is not clear that the machine can solve an
$\STCONN$ problem.

\section{Separating the $\AC^k$ Hierarchy}
\label{s:sep}

One of the obvious benefits of considering relativized complexity
classes is that separations are at hand. Even though the unrelativized
inclusion of $\ACOne$ in the polynomial hierarchy is strongly conjectured
to be strict,
no proof is currently known. On the other hand Wilson \cite{Wilson:87:mst}
showed the existence of an oracle $\alpha_W$ which makes the relativized
$\ACk$-hierarchy is strict. Here we reconstruct a technique used
by Takeuti~\cite{Takeuti:95:apal} to separate theories in weak bounded
arithmetic and use it to give a simpler definition of an oracle
$\alpha$ separating the $\ACk$ hierarchy.  In the next section we show
how to use this
result together with a witnessing theorem to obtain an unconditional
separation of relativized theories capturing the $\AC^k$
hierarchy.

The idea is that computing the $k$'th iterate
$f^k(0)=f(f(\ldots f(0)))$ of a function $f$ is essentially a
sequential procedure,
whereas shallow circuits represent parallel computation. So a circuit
performing well in a sequential task has to be deep. To avoid
the fact that the sequential character of the problem can be
circumvented by precomputing all possible values, the domain of $f$ is
chosen big enough; we will consider functions
$f\colon\{0,1\}^n\to\{0,1\}^n$.

Of course with such a big domain we cannot represent such functions
simply by a value table. That's how oracles come into
play: oracles allow us to provide a predicate on strings as input,
without the need of having an input bit for every string. In
fact, the number of bits potentially accessible by an oracle gate is
exponential in the number of its input wires.

Therefore we represent the $i$'th bit of $f(x)$ for $x\in\{0,1\}^n$ by
whether or not the string $x\underline i$ belongs to the language of
the oracle. Here $\underline i$ is some canonical coding of the
natural number $i$ using $\log n$ \footnote{We use $\log n$ to stand
for $\lceil \log_2(n+1)\rceil$.} bits.

Our argument can be summarized as follows. We assume a circuit of 
depth $d$ (i.e., the circuit has $d$ levels)
is given that supposedly computes the $\ell$'th iterate of
any function $f$ given by the oracle. Then we construct, step by step,
an oracle that fools this circuit, if $\ell>d$. To do so, for each
layer of the circuit we decide how to answer the oracle questions, and
we do this in
a way that is consistent with the previous layers and such that all
the circuit at layer $i$ knows about $f$ is at most the value of $f^i(0)$.
To make
this step-by-step construction possible we have to consider partial
functions during our construction.

%
If $A$ and $B$ are sets we denote by $f\colon A\partialto B$ that $f$
is a {partial function from $A$ to $B$}. In other words, $f$ is a
function, its domain $\dom f$ is a subset of $A$ and its range $\rng
f$ is a subset of $B$.

\begin{DEF}\label{def:ell-sequential-partial-function}
A partial function
$f\colon\{0,1\}^n\partialto\{0,1\}^n$ is called
\emph{$\ell$-sequential} if for some $k\leq\ell$ it is the case that
$0,f(0),f^2(0),\ldots,f^k(0)$ are all defined, but $f^k(0)\not\in\dom f$.
\end{DEF}

Note that
in Definition~\ref{def:ell-sequential-partial-function} it is
necessarily the case that $0,f(0),f^2(0),\ldots,f^k(0)$ are distinct.
\SHORTVERSION{For the easy proof of the next lemma,
see \cite{Aehlig:Cook:Nguyen}.}{}


\begin{lem}\label{lem:ell-sequential-extension}
Let $n\in\NN$ and $f\colon\{0,1\}^n\partialto\{0,1\}^n$ be an
$\ell$-sequential partial function.  Let $M\subset\{0,1\}^n$
be such that $|\dom f \cup M|<2^n$. Then there is an $(\ell+1)$-sequential
extension $f'\supseteq f$ with $\dom{f'}=\dom f\cup M$.
\end{lem}

\begin{proof}
Let $a\in\{0,1\}^n\setminus (M\cup\dom f)$. Such an $a$ exists by our assumption
on the cardinality of $M\cup\dom f$. Let $f'$ be $f$ extended by
setting $f'(x)=a$ for all $x\in M\setminus\dom f$. This $f'$ is as
desired. 

Indeed, assume that $0,f'(0),\ldots,f'^{\ell+1}(0),f'^{\ell+2}(0)$ are
all defined. Then, since $a\not\in\dom{f'}$, all the
$0,f'(0),\ldots,f'^{\ell+1}(0)$ have to be different from $a$. Hence these
values have already been defined in $f$. But this contradicts the
assumption that $f$ was $\ell$-sequential.
\end{proof}

\begin{DEF}\label{def:oraclefun}\label{def:Afun}
To any natural number $n$ and any partial function
$f\colon\{0,1\}^n\partialto\{0,1\}^n$ we associate its \emph{bit graph}
$\oraclefun nf$ as a partial function
$\oraclefun nf\colon\truefalse^{n+\log n}\partialto\truefalse$ in the
obvious way. More precisely, $\oraclefun nf(xv)$ is the $i$'th bit of
$f(x)$ if $f(x)$ is defined, and undefined otherwise, where 
$v$ is a string of length $\log n$ coding the natural number $i$.

If $f\colon\{0,1\}^n\to\{0,1\}^n$ is a total function, we define the
oracle $\Afun f$ by
$$
     \Afun f(w) \lra \oraclefun nf(w)=\true.
$$
\end{DEF}
Thus $\Afun f(w)$ can only hold for strings $w$ of length ${n+\log n}$.

Immediately from Definition~\ref{def:Afun} we note that $f$ can be
uniquely reconstructed from $\Afun f$.
If $\alpha$ is an oracle, we define $\Osection n \alpha$ by
$$
      \Osection n \alpha(w) \lra (\alpha(w) \wedge |w| = n+ \log n),
$$
so $\Osection n \alpha$ has finite support.

In what follows, circuits refer to oracle circuits as discussed
in Section \ref{s:cir-class}.  We are mainly interested in circuits
with no Boolean inputs, so the output depends only on the oracle.

\begin{Thm}\label{th:circuit-height-bound}
Let $C$ be any circuit of depth $d$ and size strictly less then $2^n$.
If $C(\alpha)$ correctly computes the last bit of
$f^\ell(0)$ for the (uniquely determined)
$f\colon\{0,1\}^n\to\{0,1\}^n$ such that
$\Afun f=\Osection n \alpha$, and this is true for all oracles $\alpha$,
then $\ell\leq d$.
\end{Thm}
\begin{proof}
Assume that such a circuit computes $f^\ell(0)$ correctly for all
oracles. We have to find an oracle that witnesses $\ell\leq d$. First
fix the oracle arbitrarily on all strings of length different from
$n+\log n$. So, in effect we can assume that the circuit only uses
oracle gates with $n+\log n$ inputs.

By induction on $k\ge 0$ we define partial functions
$f_k\colon\{0,1\}^n\partialto\{0,1\}^n$ with the following
properties. (Here we number the {\em levels} of the circuit
$0,1,\ldots,d-1$.)
\begin{itemize}
\item $f_0\subseteq f_1\subseteq f_2\subseteq\ldots$
\item The size $|\dom{f_k}|$ of the domain of $f_k$ is at most the
number of oracle gates in levels strictly smaller than $k$.
\item $\oraclefun n{f_k}$ determines the values of all oracle gates at levels
strictly smaller than $k$.
\item $f_k$ is $k$-sequential.
\end{itemize}
We can take $f_0$ to be the totally undefined function, since
$f^0(0) = 0$ by definition, so $f_0$ is $0$-sequential. As
for the induction step let $M$ be the set of all $x$ of length $n$
such that, for some $i<n$, the string $x\numeral i$ is queried by an oracle
gate at level $k$ and let $f_{k+1}$ be a $k{+}1$-sequential extension of $f_k$ to
domain $\dom{f_k}\cup M$ according to
Lemma~\ref{lem:ell-sequential-extension}.

For $k=d$ we get the desired bound. As $\oraclefun n{f_d}$ already
determines the values of all gates, the output of the circuit is
already determined, but $f^{d+1}(0)$ is still undefined and we can
define it in such a way that it differs from the last bit of the
output of the circuit.
\end{proof}

Inspecting the proof of Theorem~\ref{th:circuit-height-bound} we 
note that it does not at all use what precisely the
non-oracle gates compute, as long as the value only depends on the
input, not on the oracle. In particular, the proof still holds if we
consider subcircuits without oracle gates as a single complicated 
gate.
Thus we have the following corollary of the above argument and
part (ii) of Theorem \ref{t:order-b}.

\ignore{This has an important consequence, as the proof of
Theorem~\ref{t:order} shows that not only
$\csNLalpha\subseteq\AC^1(\alpha)$, but in fact, that $\csNLalpha$ can
be computed by circuits consisting of only constantly many layers of
oracle gates and oracle-independent subcircuits.
}

\begin{cor}
$\csNLalpha$ can iterate a function given by an oracle only constantly
far. In particular, there exists $\alpha$ such that $\csNLalpha$
is a strict subclass of $\AC^1(\alpha)$.
\end{cor}

Having obtained a lower bound on the depth of an individual circuit,
it is a routine argument to separate the corresponding circuit
classes. In other words, we are now interested in finding one oracle
that simultaneously witnesses that the $\ACkalpha$-hierarchy is
strict. For the uniform classes this is possible by a simple
diagonalization argument; in fact, the only property of uniformity we
need is that there are at most countably many members in each
complexity class. So we will use this as the definition of
uniformity. It should be noted that this includes all the known
uniformity notions.

\begin{DEF}\label{def:Itlang}
If $g\colon\NN\to\NN$ is a function from the natural numbers to the
natural numbers, and $\alpha$ is an oracle, we define the
language
$$\Itlang \alpha g = \{ x \mid 
\begin{array}[t]{l}
\text{the last bit of $f^{g(n)}(0)$ is $\true$,} \\
\text{where $n=|x|$ and $f$ is such that $\Afun f = \Osection n \alpha$}\}
\end{array}
$$
\end{DEF}

We note that in Definition~\ref{def:Itlang} the function $f$ is uniquely
determined by the length $n$ of $x$ and the restriction of $\alpha$
to strings of length $n+\log n$.
Also, for logspace-constructible $g$ the language $\Itlang \alpha g$ can be
computed by logspace-uniform circuits (with oracle gates)
of depth $g(n)$ and size $n\cdot g(n)$.

Recall that a {circuit family} is a sequence $\{C_n\}_{n\in\NN}$ of circuits,
such that $C_n$ has $n$ inputs and one output and may have oracle gates.
The language of a circuit family $\{C_n\}_{n\in\NN}$ with oracle $\alpha$
 is the set of all strings $x\in\truefalse^\ast$ such that the output of
 $C_{|x|}(\alpha)$ with input $x$ is 1.

\begin{DEF}
A \emph{notion of uniformity} is any countable set $\mathcal U$ of circuit
families.

Let $\mathcal U$ be a notion of uniformity, and let $d,s\colon\NN\to\NN$
be functions. The $\mathcal U$-uniform $(d,s)$-circuit families are those
circuit families $\{C_n\}_{n\in\NN}$ of $\mathcal U$ such
that $C_n$ has oracle nested depth at most $d(n)$ and size at most $s(n)$.
\end{DEF}

We use a diagonal argument to obtain the following theorem.
\begin{Thm}\label{t:separate}
Let $\mathcal U$ be a notion of uniformity and let $\{d_k\}_{k\in\NN}$
be a family of
functions such that for all $k\in\NN$ the function $d_{k+1}$
eventually strictly dominates $d_k$. Moreover, let $\{s_k\}_{k\in\NN}$
be a family of
strictly subexponentially growing functions. Then there is a single
oracle $\alpha$ that simultaneously witnesses that
for all $k$, $\Itlang \alpha {d_{k+1}}$ cannot be computed by any
$\mathcal U$-uniform $(\bigO(d_k),\bigO(s_k))$-circuit family.
\end{Thm}
 \begin{proof}
 Let $\mathcal C^0,\mathcal C^1,\ldots$ be an enumeration of $\mathcal U$.
Let $(k_i,c_i,m_i)$ be an enumeration of all triples of natural numbers.

 We will construct natural numbers $n_i$, and oracles $\alpha_i$ such that
 the following properties hold.
 \begin{itemize}
 \item The $n_i$ strictly increase.
 \item $\alpha_i(w)$ holds at most for strings $w$ of length
$n_i+\log n_i $.
 \item If $\mathcal C^{m_i}=\{C^{m_i}_n\}_{n\in\NN}$ and $C^{m_i}_{n_i}$
has oracle depth at most $c_i\cdot d_{k_i}(n_i)$ and size at most
$c_i\cdot s_{k_i}(n_i)$ then the language of
$C^{m_i}_{n_i}$ with oracle $\bigvee_{j\leq i} \alpha_j$ differs from
$\Itlang{\bigvee_{j\leq i}\alpha_j}{d_{k_i+1}}$ at some string
of length $n_i$,
and $C^{m_i}_{n_i}$ contains no oracle gates with $n_{i+1}+\log n_{i+1} $
or more inputs.
\end{itemize}
Then $\bigvee_i\alpha_i$ will be the desired oracle $\alpha$.  For
suppose otherwise.  Then there is $k$ and a $(\bigO(d_k),\bigO(s_k))$ circuit
family $\mathcal C^{m}$ such that $\mathcal C^{m}$ with oracle $\alpha$
computes  $\Itlang \alpha {d_{k+1}}$.
Hence there is a triple $(k_i,c_i,m_i)$ such
that the circuit $C^{m_i}_{n_i}$ has oracle depth at most
$c_i\cdot d_{k_i}(n_i)$ and size at most $c_i\cdot s_{k_i}(n_i)$ and
$C^{m_i}_{n_i}(\alpha)$ correctly computes $\Itlang{\alpha}{d_{k_i+1}}$
on inputs of length $n_i$.  But this contradicts the above properties.

At stage $i$ take $n_i$ big enough so that it is bigger than all the
previous $n_j$'s and $n_i+\log n_i $ is bigger than the maximal
fan-in of all the oracle gates in all the circuits looked at so far;
moreover take $n_i$ big enough such that $c_i\cdot s_{k_i}(n_i)<2^{n_i}$
and $c_i\cdot d_{k_i}(n_i)<d_{k_i+1}(n_i)$. 
This is possible as $s_{k_i}$ has strictly subexponential growth and
$d_{k_i+1}$ dominates $d_{k_i}$ eventually.

Look at the $n_i$'th circuit in the circuit family $\mathcal C^{m_i}$,
and call it $C$. We may assume that $C$ has oracle depth at most
$c_i\cdot d_{k_i}(n_i)$ and size at most $c_i\cdot s_{k_i}(n_i)$ for
otherwise there is nothing to show and we can choose $\alpha_i$ to be
the empty set.

By Theorem~\ref{th:circuit-height-bound} we can find an oracle $\alpha_i$
whose support includes only strings of length $n_i+\log n_i $ such that
$C$ with oracle $\bigvee_{j\leq i}\alpha_i$ does not solve the
decision problem associated
with $f^{d_{k_i+1}(n_i)}(0)$ for $f$ given by $\Afun f= \alpha_i$.
\end{proof}

\begin{cor}\label{c:singlestrict}
There is a single oracle $\alpha \subseteq\truefalse^\ast$ for which 
\begin{equation}\label{e:strict}
   \AC^k(\alpha) \subseteq \NC^{k+1}(\alpha) \subsetneq \AC^{k+1}(\alpha),
  \mbox{ for all $k\ge 1$}
\end{equation}
and $\csNLalpha\subsetneq \AC^1(\alpha)$.
\end{cor}
\begin{proof}
In Theorem \ref{t:separate} let $\mathcal U$ be log space uniformity
and let $d_k(n) = \log^k n$ and $s_k(n) = 2^{\lceil n/2\rceil}$.  
Then for $k\ge 1$, every problem in $\AC^k(\alpha)$ can be
computed by a $\mathcal U$-uniform $(\bigO(d_k),\bigO(s_k))$-circuit family
with oracle $\alpha$, and $\Itlang \alpha {d_{k}}$ is in
$\AC^k(\alpha)$.  Then Theorem \ref{t:separate} shows that there
is a single oracle $\alpha$ satisfying (\ref{e:strict}).

To show $\csNLalpha\subsetneq \AC^1(\alpha)$, by Theorem
\ref{t:order-b} part (ii), it suffices to show
\begin{equation}\label{e:suffice}
\Itlang \alpha {d_{1}} \notin \ACZ(\STCONN,\alpha).
\end{equation}
Theorem \ref{t:separate} shows how to construct $\alpha$ so
$\Itlang \alpha {d_{1}} \notin \ACZ(\alpha)$, and we can modify
the proof by starting with $\alpha$ which is a version of $\STCONN$
in which no string has length $n + \log n$ for any $n$.
The result of the construction in the proof satisfies (\ref{e:suffice})
as well as (\ref{e:strict}).
\end{proof}


\section{Theories for Relativized Classes}
\label{s:theory}

\subsection{Two-Sorted Languages and Complexity Classes }

Our theories are based on a two-sorted vocabulary, and it is
convenient to re-interpret the complexity classes using this
vocabulary \cite{Cook:Nguyen,Nguyen:Cook:05:lmcs}.
Our two-sorted language has variables
$x,y,z,...$ ranging over $\N$ and variables $X,Y,Z,...$ ranging over
finite subsets of $\N$ (interpreted as bit strings).  Our basic two-sorted
vocabulary \LTwoA\ includes the usual symbols
$0,1,+,\cdot,=,\le$ for arithmetic over $\N$, the length function
$|X|$ on strings, the set membership relation $\in$, and string
equality $=_2$ (where we usually drop mention of the subscript 2).
The function $|X|$ denotes 1 plus the largest element in the set $X$,
or 0 if $X$ is empty (roughly the length of the corresponding string).
We will use the notation $X(t)$ for $t\in X$, and we will think of $X(t)$ as
the $t$-th bit in the string $X$.

{\em Number terms} of \LTwoA\
are built from the constants 0,1, variables $x,y,z,...$,
and length terms $|X|$, using $+$ and $\cdot$.  The only {\em string terms}
are string variables $X,Y,Z,...$.   The atomic formulas are
$t=u$, $X=Y$, $t\le u$, $t\in X$ for any number terms $t,u$ and string
variables $X,Y$.  Formulas are built from atomic formulas using
$\wedge,\vee, \neg$ and both number and string quantifiers
$\exists x, \exists X, \forall x,\forall X$.  Bounded number quantifiers
are defined as usual, and the bounded string quantifier
\mbox{$\exists X \le t \  \varphi$} stands for $\exists X(|X|\leq t \wedge \varphi)$
and $\forall X\le t \ \varphi$ stands for $\forall X(|X|\le t\supset \varphi)$,
where $X$ does not occur in the term $t$.

\SigZB\ is the set of all \LTwoA-formulas in which all number quantifiers
are bounded and with no string quantifiers.  \SigOneB\
(corresponding to {\em strict} $\Sigma^{1,b}_1$ in \cite{Krajicek:95:book})
formulas begin with
zero or more bounded existential string quantifiers, followed by a \SigZB\
formula.  These classes are extended to \SigIB, $i\ge 2$,
(and \PiIB, $i\ge 0$) in the usual way.

We use the notation \SigZB(\calL) to denote \SigZB\ formulas which
may have two-sorted function and predicate symbols from the vocabulary
\calL\ in addition to the basic vocabulary \LTwoA.

Two-sorted complexity classes contain relations $R(\xvec,\Xvec)$
(and possibly number-valued functions $f(\xvec,\Xvec)$ or
string-valued functions $F(\xvec,\Xvec)$), where the arguments
$\xvec = x_1,\ldots,x_k$ range over $\N$, and $\Xvec = X_1,\ldots,X_\ell$
range over finite subsets of $\N$.  In defining complexity classes
using machines or circuits, the number arguments $x_i$ are presented in unary
notation (a string of $x_i$ ones), and the arguments $X_i$ are
presented as bit strings.  Thus the string arguments are the important
inputs, and the number arguments are small auxiliary inputs
useful for indexing the bits of strings.

As mentioned before, uniform \ACZ\ has several equivalent
characterizations \cite{Immerman:99:book}, including \LTH\
(the log time hierarchy on alternating Turing machines) and \FO\
(describable by a first-order formula using predicates for
plus and times).  Thus in the two-sorted setting we can
define \ACZ\ to be the class of relations
$R(\xvec,\Xvec)$ such that some alternating Turing
machine accepts $R$ in time $O(\log n)$ with a constant number of
alternations, using the input conventions for numbers and strings
given above.  Then from the \FO\ characterization of \ACZ\
we obtain the following nice connection between the classes \ACZ\ and
\NP\ and our two-sorted \LTwoA-formulas (see Theorems IV.3.6 and IV.3.7
in \cite{Cook:Nguyen}).

\begin{theorem}[Representation Theorem]\label{t:repres}
\label{t:repre}
A relation $R(\xvec,\Xvec)$ is in \ACZ\ (resp. \NP) iff it is
represented by some \SigZB\ (resp. \SigOneB) formula
$\varphi(\xvec,\Xvec)$.
\end{theorem}

In general, if {\bf C} is a class of relations (such as \ACZ) then
we want to associate a class {\bf FC} of functions with {\bf C}.  Here
{\bf FC} will contain string-valued functions $F(\xvec,\Xvec)$ and
number-valued functions $f(\xvec,\Xvec)$.  We require that
these functions be $p$-bounded; i.e. for each $F$ and $f$ there
is a polynomial $g(n)$ such that $|F(\xvec,\Xvec)|\le
g(\max(\xvec,|\Xvec|))$ and $f(\xvec,\Xvec) \le g(\max(\xvec,|\Xvec|))$.

We define the {\em bit graph} $B_F(i,\xvec,\Xvec)$ to hold iff
the $i$th bit of $F(\xvec,\Xvec)$ is one.  Formally
\begin{equation}\label{d:bit-graph}
   B_F(i,\xvec,\Xvec) \lra F(\xvec,\Xvec)(i)
\end{equation}
(Compare this with Definition \ref{def:Afun}.)

\begin{definition}
If {\bf C} is a two-sorted complexity class of relations, then
the corresponding function class {\bf FC} consists of all p-bounded
number functions whose graphs are in {\bf C}, together with all
p-bounded string functions whose bit graphs are in {\bf C}.
\end{definition}

For example, binary addition $F_+(X,Y) = X+Y$ is in \FACZ,
but binary multiplication $F_\times(X,Y) = X\cdot Y$ is not.

\begin{definition}\label{d:SigZB-def}
A string function is {\em $\SigZB$-definable} from a collection $\calL$
of two-sorted functions and relations if it is p-bounded and its bit graph
is represented by a $\SigZB(\calL)$ formula.
Similarly, a number function is {\em $\SigZB$-definable} from $\calL$ if
it is p-bounded and its graph is represented by a $\SigZB(\calL)$ formula.
\end{definition}

It is not hard to see that \FACZ\ is closed under \SigZB-definability,
meaning that if the bit graph of $F$ is represented by a \SigZB(\FACZ)
formula, then $F$ is already in \FACZ.   Of course the set of functions
in \FACZ\ is closed under composition, but for a general vocabulary
$\calL$ the set of functions $\SigZB$-definable from $\calL$ may
not be closed under composition.  For example if a relation $\alpha(Y)$
codes the bit graph of a function $F$ then $F$
could be $\SigZB$-definable from $\calL\cup\{\alpha\}$ but $F\circ F$
may not be.  In order to define
complexity classes such as $\ACZ(m)$ and \TCZ, as well as relativized
classes such as \ACZalpha, we need to iterate
\SigZB-definability to obtain the notion of \ACZ\ reduction.

\begin{definition}\label{d:reducible}
We say that a string function $F$ (resp. a number function $f$)
is \ACZ-reducible to $\calL$ if there is a
sequence of string functions $F_1, \ldots, F_n$ ($n \ge 0$) such that
\begin{equation}
\label{e:AC0-red}
F_i
\text{ is \SigZB-definable from }\calL \cup \{F_1, \ldots, F_{i-1}\},
\text{ for }i = 1, \ldots, n;
\end{equation}
and $F$ (resp. $f$) is $\SigZB$-definable from
$\calL \cup \{F_1, \ldots, F_n\}$.
A relation $R$ is $\ACZ$-reducible to $\calL$ if there is a
sequence $F_1, \ldots, F_n$ as above, and $R$ is represented by a
$\SigZB(\calL \cup \{F_1, \ldots, F_n\})$ formula.
\end{definition}

If $F$ and $G$ are string-to-string functions in $\calL$ then
the term $F(G(X))$ can appear in $\SigZB(\calL)$ formulas,
so the set of functions \ACZ-reducible to $\calL$ is always closed
under composition.  In fact from the techniques used to prove
Theorem \ref{t:repre} we can show that
$F$ is \ACZ-reducible to $\calL$ iff there is a uniform
constant-depth polysize circuit family that computes $F$, where
the circuits are allowed gates (each of depth one) which
compute the functions and predicates in $\calL$ (as well as the
Boolean connectives).

The (two-sorted) classes $\ACZ(m), \TCZ, \NCOne, \L$ and $\NL$ are the closure under \ACZ-reductions
of their respective complete problems, so they become
$\ACZ(\modm)$, $\ACZ(\THRESH)$, $\ACZ(\FormVal)$, $\ACZ(\OneSTCONN)$
and $\ACZ(\STCONN)$.
The relativized versions $\ACZ(\alpha)$, $\ACZ(m,\alpha)$, $\TCZ(\alpha)$,
$\NC^k(\alpha)$, and $\AC^k(\alpha)$ of the circuit classes are all
closed under $\ACZ$-reductions, and so is $\csL(\alpha)$
(Theorem \ref{t:order-b} part (i)), but $\csNLalpha$ may not be so closed.


\subsection{Nonrelativized Theories}\label{s:nonrel-theories}

In this paper we consider theories $\calT$ over two-sorted vocabularies
which contain \LTwoA.  

\begin{definition}\label{d:definable}
If $F(\xvec,\Xvec)$ is a string function,
we say that $F$ is \SigOneB-{\em definable} (or {\em  provably total})
in \calT\ if there is a \SigOneB\ formula $\varphi(\xvec,\Xvec,Y)$ which
represents the graph of $F$ and 
$$
    \calT \vdash \exists! Y \varphi(\xvec, \Xvec,Y)
$$
where  $\exists! Y$ means there exists a unique $Y$.
\end{definition}

A similar definition applies to for number functions $f(\xvec,\Xvec)$.
When we associate a theory \calT\ with a complexity class $\boldC$
we want the provably total functions in \calT\ to coincide with
the functions in $\FC$. 

The theory \VZ\ (essentially $\SigZp$-{\it comp} in \cite{Zambella:96:jsl},
and $\mathbf{I\Sigma}_0^{1,b}$ (without \#) in \cite{Krajicek:95:book})
is the theory over \LTwoA\ that is axiomatized
by the axioms listed in Figure \ref{d:2-BASIC} together with the 
comprehension axiom scheme \COMP{\SigZB},
i.e. the set of all formulas of the form
\begin{equation} \label{e:COMP}
   \exists X\leq y\forall z<y(X(z)\lra \varphi(z)),
\end{equation}
where $\varphi(z)$ is any formula in $\SigZB$,
and $X$ does not occur free in $\varphi(z)$.

\begin{figure}
\centering
\begin{tabular}{|ll|}
\hline
{\bf B1.} $x+1\neq 0$                           & {\bf B7.} $(x\le y \wedge y\le x)\supset x = y$\\
{\bf B2.} $x+1 = y+1\supset x=y$                & {\bf B8.} $x \le x + y$\\
{\bf B3.} $x+0 = x$                             & {\bf B9.} $0 \le x$\\
{\bf B4.} $x+(y+1) = (x+y) + 1$                 & {\bf B10.} $x\le y \vee y\le x$\\
{\bf B5.} $x\cdot 0 = 0$                        & {\bf B11.} $x\le y \lra x< y+1$ \\
{\bf B6.} $x\cdot (y+1) = (x\cdot y) + x$       & {\bf B12.} $x\neq 0 \supset \exists y\le x(y+1 = x)$\\
{\bf L1.} $X(y) \supset y < |X|$                & {\bf L2.}  $y+1 = |X| \supset X(y)$ \\
\multicolumn{2}{|l|}
{{\bf SE.} $[|X| = |Y|\wedge\forall i<|X|(X(i)\lra Y(i))]\ \supset\ X = Y$}\\
\hline
\end{tabular}
\caption{2-{\bf BASIC}}
\label{d:2-BASIC}
\end{figure}

We associate $\VZ$ with the complexity class \ACZ, and indeed the
provably total functions of \VZ\ comprise the class \FACZ.
All theories considered in this paper extend \VZ.


In \cite[Chapter IX]{Cook:Nguyen},
for various subclasses \boldC\ of \Ptime, a theory \VC\ is developed
which is associated with \boldC\ as above, so the provably total
functions of \VC\ are precisely those in \FC.
Essentially, the theory \VC\ is axiomatized by the axioms of \VZ\ together with
an axiom that states the existence of a polytime computation for a
complete problem of \boldC, assuming the parameters as given inputs.
The additional axioms for the classes of interest in this paper will be listed below.
For the logspace classes (i.e., $\ACZm,\ldots, \NL$) we use roughly the same
problems as those mentioned in the previous sections.
For other classes in the \AC\ hierarchy we use the monotone circuit value problem,
with appropriate restrictions on the depth and fanin of the circuits.

To formulate these axioms we introduce the pairing function $\tuple{y,z}$,
which stands for the term $(y+z)(y+z+1)+2z$. This allows us to interpret
a string $X$ as a two-dimensional bit array, using the notation
\begin{equation}\label{e:bit-array}
   X(y,z) \equiv X(\tuple{y,z})
\end{equation}
For example, a graph with $a$ vertices can be encoded by a pair $(a,E)$ where
$E(u,v)$ holds iff there is an edge from $u$ to $v$, for $0\le u, v < a$.

We will also use the number function $(Z)^x$ which is the $x$-th element of the sequence of numbers encoded by $Z$:
\begin{multline*}
y=(Z)^x \lra (y < |Z| \wedge Z(x,y) \wedge \forall z < y \neg Z(x,z))
 \vee\\
(\forall z < |Z| \neg Z(x,z) \wedge y = |Z|)
\end{multline*}
In addition, $\log{a}$, or $|a|$, denotes the integral part of
$\log_2(a+1)$.
Note that these function is provably total in \VZ.
Note also that the functions $(Z)^x$ and $|a|$ can be eliminated using their \SigZB\ defining axioms,
so that \SigZB\ formulas that contain these functions are in fact equivalent to \SigZB\ formulas over \LTwoA\
(see \cite[Lemma V.4.15]{Cook:Nguyen}).

We now list the additional axioms for the classes considered in this paper.
(Recall that the base theory is always \VZ.)
First, consider \TCZ.
The theory \VTCZ\ is axiomatized by the axioms of \VZ\ and the following axiom:
$$\NUMONES \equiv  \exists Y \le 1 + \tuple{x,x} \delNum(x,X,Y)$$
where
\begin{multline*}
\delNum(x,X,Y)\equiv
 \ (Y)^0 = 0\ \wedge\\
 \forall z<x\, \left[(X(z) \supset (Y)^{z+1} = (Y)^z + 1) \wedge (\neg
 X(z) \supset (Y)^{z+1} = (Y)^z)\right]
\end{multline*}
Here \NUMONES\ formalizes a computation of the number of 1-bits in the string $X(0), X(1),\ldots, X(x-1)$:
for $1\le z \le x$, $(Y)^z$ is the number of 1 bits in $X(0), X(1), \ldots, X(z-1)$.
(Note that computing the number of 1-bits in the input string is roughly the same as computing
the threshold function.)

Now consider \ACZm.
The additional axiom for the theory \VZm\ associated with \ACZm\ is
$$
\MOD_m \equiv  \exists Y \delta_{\MOD_m}(x,X,Y)
$$
where
\begin{multline*}
\delta_{\MOD_m}(x,X,Y)\equiv
Y(0,0)\ \wedge\\ 
 \forall z<x\,\left[(X(z)\supset (Y)^{z+1} = (Y)^z + 1\mod{m})\wedge(\neg X(z) \supset (Y)^{z+1} = (Y)^z)\right]
\end{multline*}
Similar to \NUMONES\ above,
here $(Y)^z$ is the number of 1-bits in the sequence
$X(0), \ldots, X(z-1)$ modulo $m$.

For \NCOne, the additional axiom \MFV\ for \VNCOne\ states the existence of a polytime computation
for the Balanced Monotone Sentence Value Problem which is complete
for \VNCOne\ \cite{Buss:87:stoc}:
\begin{equation}\label{e:MFV}
\MFV \ \equiv\  \exists Y\, \delMfvp(a,G,I,Y)
\end{equation}
where
\begin{multline*}
\delMfvp(a,G,I,Y) \equiv \
\forall x<a\,[Y(x+a) \lra I(x) \wedge\\
	0 < x\ \supset\ Y(x) \lra \left(\big(G(x) \wedge Y(2x) \wedge Y(2x+1)\big) \vee \right.\\
        \left. \big(\neg G(x) \wedge (Y(2x) \vee Y(2x+1))\big)\right)
\end{multline*}
Here the balanced monotone sentence is viewed as a balanced binary tree encoded by $(a,G)$
and $I$ specifies the leaves of the tree: node $x$'s children are $2x$ and $2x+1$,
$G(x)$ indicates whether node $x$ is an $\vee$ or $\wedge$ node, and $I(z)$ is the
value of leaf $z$.
$Y$ is the bottom-up evaluation of the sentence: $Y(x)$ is the value of node $x$.

Next, the theory \VNL\ for \NL\ is axiomatized by \VZ\ and the following axiom
$$\CONN \equiv  \exists Y\, \delConn(a, E, Y)$$
where $\delConn(a,E,Y)$ states that $Y$ encodes a polytime computation for the following problem,
which is equivalent to \STCONN\ under \ACZ-many-one reductions:
on input $(a,E)$ which encodes a directed graph, compute the nodes that are reachable from node 0.
Here $Y(z,x)$ holds iff there is a path from 0 to $x$ of length $\le z$;
more precisely,
\begin{multline*}
\delConn(a,E,Y) \equiv Y(0,0) \wedge \forall x < a (x \neq 0 \supset \neg Y(0,x))\ \wedge \\
 \forall z < a \forall x<a\, \left[Y(z+1, x) \lra (Y(z,x) \vee \exists
 y<a,\ Y(z,y) \wedge E(y,x))\right].
\end{multline*}

For \L, the additional axiom in \VL\ is \PATH, which states the existence of a polytime computation
for the following problem, which is equivalent to \OneSTCONN\ under \ACZ\ reductions:
on input a directed graph with out degree at most one which is encoded by $(a,E)$,
compute the transitive closure of vertex 0:
\begin{multline*}
\PATH \ \equiv\ \forall x < a \exists! y < a E(x,y)\ \supset\\
\exists P\, \left[(P)^0 = 0 \wedge \forall v < a
\big((P)^{v+1} < a \wedge E((P)^v,(P)^{v+1})\big)\right]
\end{multline*}
($(P)^v$ is the vertex of distance $v$ from $0$.)

Now we consider the classes \ACk\ for $k \ge 1$.
We use the monotone circuit value problem under an appropriate setting.
In particular, the circuit has unbounded fanin, and its depth is
$(\log n)^k$, where $n$ is the number of its inputs.
Thus the additional axiom in \VACk\ states the existence of a polytime evaluation $Y$
for a circuit of this kind which is encoded by $(a,E, G, I)$:
\begin{equation}\label{e:axVACk}
 \exists Y\, \delLmcv(a,|a|^k,E,G,I,Y)
\end{equation}
where the depth parameter $d$ is set to $|a|^k$ in the formula
\begin{multline*}
\delLmcv(w,d,E,G,I,Y) \equiv \forall x < w \forall z < d,\ (Y(0,x) \lra I(x))\ \wedge  \\
Y(z+1, x) \lra
        \big[[G(z+1,x) \wedge \forall u < w\,(E(z, u, x) \supset Y(z,u))] \vee\\
        [\neg G(z+1,x) \wedge \exists u < w\,( E(z, u, x) \wedge
        Y(z,u))]\big]
\end{multline*}
The formula $\delLmcv(w,d,E,G,I,Y)$ (Layered Monotone Circuit Value)
states that $Y$ is an evaluation of the
circuit encoded by $(w,d,E,G)$ on input $I$.
The circuit is encoded as follows.
There are $(d+1)$ layers in the circuit, each of them contains $w$ gates.
Hence each gate is given by a pair $(z,x)$ where $z$ indicates the layer
(inputs to the circuits are on layer 0 and outputs are on layer $d$),
and $x$ is the position of the gate on that layer.
$E$ specifies the wires in the circuit: $E(z,u,x)$ holds if and only if gate $(z,u)$
is an input to gate $(z+1,x)$,
and $G$ specifies the gates: $G(z,x)$ holds if gate $(z,x)$ is an $\wedge$ gate,
otherwise it is an $\vee$ gate.
Bit $Y(z,x)$ is the value of gate $(z,x)$.


For \NCk\ ($k \ge 2$) the circuit value problem is restricted further, so that the circuit's fanin is at most 2.
We express this condition by the formula $\Fanin(w,d,E)$,
where $(w,d,E)$ encodes the underlying graph of the circuit as above:
\begin{multline*}
\Fanin(w,d,E)\ \equiv\ \forall z < d \forall x < w\exists u_1 < w \exists u_2 < w \forall v < w\\ 
	E(z,v,x) \supset (v = u_1 \vee v = u_2)
\end{multline*}
Similar to \VACk, the theory \VNCk\ (for $k\ge 2$) is axiomatized by the axioms of \VZ\ together with
\begin{equation}\label{e:axVNCk}
(\Fanin(a,|a|^k,E) \supset \exists Y\, \delLmcv(a,|a|^k,E,G,I,Y))
\end{equation}


The connection between the above theories \VACk\ and \VNCk\
and their corresponding classes 
is discussed in detail in \cite[Section IX.5.6]{Cook:Nguyen}.
The key point for these theories (as well as the others) is that the
problem of witnessing the existential
quantifiers in the axiom for each theory
(in this case (\ref{e:axVACk})   and (\ref{e:axVNCk})) is
complete for the associated complexity class.
For \ACk\ and \NCk\ this is shown by using the characterization of these
classes in terms of alternating Turing machines.

In the remainder of this section we will present relativized theories $\VC(\alpha)$
that characterize the relativized classes discussed in Section \ref{s:class}.

\subsection{Relativized Theories}\label{s:rel-theories}

\subsubsection{Classes with bounded nested oracle depth}
\label{s:bounded_depth}

We first look at
$\ACZ(m,\alpha)$, $\TCZ(\alpha)$, $\NCOne(\alpha)$ and $\csLalpha$.
These classes have constant nested depth of oracle gates,
and they are the \ACZ-closure of the oracle $\alpha$ and an appropriate complete problem
(see Proposition \ref{p:complete} and Theorem \ref{t:order-b} (i)).
We can treat these as the $\ACZ(\alpha)$-closure of the complete problem
for their respective nonrelativized version.
Thus the development in \cite[Chapters VIII and IX]{Cook:Nguyen} can be readily extended to these classes.
The change we need to make here is to replace the base theory \VZ\ by its relativized version, 
$\VZ(\alpha)$, which is axiomatized by comprehension axioms
(\ref{e:COMP}) over
$\SigZB(\alpha)$ formulas instead of just $\SigZB$ formulas.

First note that a sequence of strings can be encoded using the
string function $\Row$, where $\Row(x,Z)$ extracts row $x$ from
the array coded by $Z$.  Thus
\begin{equation}\label{e:Rowdefine}
   \Row(x, Z)(i)\lra i<|Z|\wedge Z(x,i)
\end{equation}
We will also write $Z^{[x]}$ for $\Row(x,Z)$.

\Notation.
For a predicate $\alpha$ we use $\LTwoA(\alpha)$ to denote
$\LTwoA\cup\{\Row,\alpha\}$,
and we use $\SigZBalpha$ and $\SigOneB(\alpha)$
to denote the classes  $\SigZB(\Row,\alpha)$ and $\SigOneB(\Row,\alpha)$,
respectively.   Definitions \ref{d:SigZB-def}, \ref{d:reducible},
and \ref{d:definable}
of \SigZB-definable, \ACZ-reducible, and \SigOneB-definable in a
theory, are extended in the obvious way to
$\SigZB(\alpha)$-definable, $\ACZ(\alpha)$-reducible, and
$\SigOneB(\alpha)$-definable in a theory.

Atomic formulas containing $\alpha$ have the form $\alpha(T)$,
where $T$ is a string term; namely either a variable $X$ or a term
$\Row(t, T')$ for terms $t,T'$.

Notice that while the string function $\Row$ can occur nested in a $\SigZB(\alpha)$ formula,
the predicate $\alpha$ cannot.
Thus a $\SigZB(\alpha)$ formula represents relations computable by a family
of polynomial size $\ACZ(\alpha)$ circuits whose oracle nested depth is one.

The function $\Row$ is useful in constructing formulas describing
circuits which query the oracle $\alpha$.  For example if an $n$-ary
gate $g$ has inputs from $n$ different $\alpha$ gates, we can code the
sequence of inputs to the $\alpha$ gates
using a string $X$, so the $i$th input bit to $g$ is $\alpha(X^{[i]})$.

\ignore{
The next theorem can be proved in the same way as Theorem \ref{t:repre}.

\begin{theorem}
A relation $R(\xvec,\Xvec)$ is in $\ACZ(\alpha)$ iff it is represented by some
$\SigZBalpha$ formula $\varphi(\xvec,\Xvec)$.
\end{theorem}
}


\begin{definition}\label{d:relTheories}
The following theories have vocabulary $\LTwoA(\alpha)$
and include the defining axiom (\ref{e:Rowdefine}) for \Row.
$\VZalpha = \VZ + \COMP{\SigZBalpha}$.
The theories $\VZ(m,\alpha)$, $\VTCZalpha$, $\VNCOnealpha$ and $\VLalpha$,
$\VNLprimealpha$
are axiomatized by the axioms of \VZalpha\ together with the axiom:
$\MOD_m$, $\NUMONES$, $\MFV$, $\PATH$, $\CONN$ respectively.
(See Section \ref{s:nonrel-theories}.)
\end{definition}

Note that equality axioms (implicitly) hold for the new symbols $\Row$ and $\alpha$.



The next result connects the theories with their corresponding
complexity classes, except for the theory $\VNLprimealpha$, which
corresponds to the class $\ACZ(\STCONN,\alpha)$ (see part (ii) of
Theorem \ref{t:order-b}) rather than $\csNLalpha$.  We are not able
to provide a theory exactly associated with $\csNLalpha$ because we
cannot show that the associated function class is closed under
composition.

The first step in the proof of the next theorem
is to show that a function is
in $\FACZ(\alpha)$ iff it is $\SigOneB(\alpha)$-definable in \VZalpha.
For the direction $(\Rightarrow)$  the proof of the unrelativized case
uses the \SigZB\ Representation Theorem (Theorem \ref{t:repres}), but
that result does not hold for the relativized case, because,
as remarked above, $\SigZB(\alpha)$ formulas only represent
$\ACZ(\alpha)$ relations
that can be computed by circuits of oracle depth one.

So results in Chapter IX of \cite{Cook:Nguyen} are required.

\begin{theorem}\label{t:def-VNL}
For a class \boldC\ in $\{\ACZ, \ACZm, \TCZ, \NCOne, \L\}$,
a function is in \FCalpha\ if and only if it is
$\SigOneB(\alpha)$-definable in \VCalpha.
\end{theorem}

\begin{proof}
We start by proving this when \boldC\ is \ACZ:  A function is in
$\FACZ(\alpha)$ iff it is $\SigOneB(\alpha)$-definable in \VZalpha.
The `if' direction follows from a standard witnessing theorem
(see for example Chapter V in \cite{Cook:Nguyen}), because
the existential quantifier in each \COMP{\SigZBalpha} axiom is
witnessed by an $\FACZ(\alpha)$ function whose graph is represented by
a \SigZBalpha\ formula.

The converse follows from a slight generalization of
Theorem IX.2.3 in \cite{Cook:Nguyen}, where the original
states that a function is $\SigOneB$-definable in \VC\ iff it is
in \FC.  That theorem applies to complexity classes \boldC\ consisting
of the relations \ACZ-reducible to a string function $F(X)$ whose
graph is represented by a \SigZB-formula $\delta_F(X,Y)$.  The theory \VC\
has vocabulary \LTwoA\ and
is axiomatized by the axioms of \VZ\ together with
\begin{equation}\label{e:VCax}
    \exists Y \le b \forall i<b \delta_F(X^{[i]},Y^{[i]}).
\end{equation}
The generalization we need (which is proved in the same way) is that
the theory \VC\ is replaced by
a theory $\VC(\alpha)$ with vocabulary $\LTwoA(\alpha)$
axiomatized by the axioms of $\VZ(\alpha)$ and (\ref{e:VCax}), where
now $\delta_F(X,Y)$ is a $\SigZB(\alpha)$-formula.  The assertion now
is that a function is $\SigOneB(\alpha)$-definable in $\VC(\alpha)$ iff
it is in $\FC(\alpha)$.

To apply this to the theory \VZalpha\ we take $F = F_\alpha$ where
$\delta_{F_\alpha}(X,Y)$ is the formula
$|Y| \le |X| \wedge\forall j<|X|(Y(j) \lra \alpha(X^{[j]}))$.
Thus $F_\alpha(X)$ is the bit string resulting from applying $\alpha$
successively to the elements of the sequence of strings coded by $X$,
and so the \ACZ\ closure of $F_\alpha$ is $\FACZ(\alpha)$.
The theory \VZalpha\ has the comprehension axiom $\COMP{\SigZBalpha}$,
which implies (\ref{e:VCax}) when $F$ is taken to be $F_\alpha$.
Thus our generalized Theorem IX.2.3 in \cite{Cook:Nguyen} implies that 
every function
in $\FACZ(\alpha)$ is $\SigOneB(\alpha)$-definable in \VZalpha.

Theorem \ref{t:def-VNL} for the other complexity classes follows
from Proposition \ref{p:complete} and Theorem \ref{t:order-b} (i)
and the fact that the theories for the nonrelativized classes
capture the nonrelativized classes (Section \ref{s:nonrel-theories}).
\end{proof}

The same argument shows that the theory $\VZ(\alpha) + \CONN$ is
associated with the class $\ACZ(\STCONN,\alpha)$. But this class
might not be the same as $\csNLalpha$.
In fact, as pointed out after the proof of Corollary \ref{c:cslComp},
the function class associated with $\csNLalpha$ may not be closed under
composition, and hence $\csNLalpha$ may not
be closed under \ACZ-reductions,
so the framework of \cite[Chapter IX]{Cook:Nguyen} may not apply to this
class.

Notice also that we can relativize the axioms $\MOD_m$,
$\NUMONES$, $\MFV$, and $\PATH$ in the obvious way,
i.e., by replacing the string variables $X$ in \NUMONES\ and $\MOD_m$,
$G$ and $I$ in \MFV, and $E$ in \PATH\ by $\SigZB(\alpha)$ formula(s).
It turns out that these relativized axioms are provable in the respective relativized theories,
and in fact they can be used together with \VZ\ to axiomatize the theories.
More specifically, let
$\NUMONES(\alpha)$ denote the following {\em axiom scheme}:
\begin{multline}\label{e:NUMONESalpha}
\exists Y \le 1 + \tuple{x,x} \ (Y)^0 = 0\ \wedge\\
 \forall z<x\, \left[(\varphi(z) \supset (Y)^{z+1} = (Y)^z + 1) \wedge (\neg
 \varphi(z) \supset (Y)^{z+1} = (Y)^z)\right]
\end{multline}
for all $\SigZB(\alpha)$ formulas $\varphi$ that do not contain $Y$.
Similarly we can define $\MOD_m(\alpha)$, $\MFV(\alpha)$,
and $\PATH(\alpha)$.

The next result is useful in the next subsection.

\begin{proposition}\label{t:VTC-alpha-alt}
\VTCZalpha\ can be equivalently axiomatized by the axioms of \VZ\ and
$\NUMONES(\alpha)$.  Similarly for 
$\VZ(m,\alpha)$, $\VNCOnealpha$, and $\VLalpha$,
with $\NUMONES(\alpha)$ replaced respectively by
$\MOD_m(\alpha)$, $\MFV(\alpha)$, and $\PATH(\alpha)$.
\end{proposition}

\begin{proof}
We prove this for $\NUMONES(\alpha)$.  The other cases are similar.

It is relatively simple to show that the axioms of $\NUMONES(\alpha)$ are provable in $\VTCZalpha$.
Indeed, consider an axiom in $\NUMONES(\alpha)$ as in \eqref{e:NUMONESalpha} above.
By \COMP{\SigZB(\alpha)}, there is a string $X$ such that $X(z) \lra \varphi(z)$ for all $z < x$.
Hence, the string $Y$ that satisfies $\delNum(x,X,Y)$ satisfies \eqref{e:NUMONESalpha}.

For the other direction, suppose that we want to prove the following instance of 
\COMP{\SigZB(\alpha)} using $\NUMONES(\alpha)$ and \VZ:
$$\exists Z \le b \forall z < b,\ Z(z) \lra \varphi(z)$$
where $\varphi$ is a $\SigZB(\alpha)$ formula.
Using $\NUMONES(\alpha)$ we obtain a string $Y$ as in \eqref{e:NUMONESalpha} (for $x = b$).
Now, it is straightforward to identify those $z < b$ such that $\varphi(z)$ holds:
$$\varphi(z) \lra (Y)^{z+1} = (Y)^{z}+1$$
Thus $Z$ can be defined using $\COMP{\SigZB}$ from $Y$.

The arguments for $\VZ(m,\alpha)$, $\VNCOnealpha$, and $\VLalpha$ are similar.
\end{proof}

\ignore{
Notice that natural relativized versions of the additional
axioms of \VC, such as \CONN, are already provable in $\VC(\alpha)$.
For example, let $\CONNalpha$ be the axiom scheme
\begin{multline*}
\forall a\exists Y\ [Y(0,0) \wedge \forall x < a (x \neq 0 \supset \neg Y(0,x))\ \wedge \\
 \forall z < a \forall x<a,\ Y(z+1, x) \lra (Y(z,x) \vee \exists y<a,\
 Y(z,y) \wedge \varphi(y,x))].
\end{multline*}
where $\varphi$ is a \SigZBalpha\ formula.
Then it is easy to use \COMP{\SigZBalpha} to show that
$\VNL(\alpha) \vdash \CONNalpha$.
}

\ignore{
Details can be found in \cite{Aehlig:Cook:Nguyen}.
\begin{proof}
Consider the ($\Longra$) direction for $\boldC = \ACZ$.  Here
$\FACZ(\alpha)$ consists of all $p$-bounded functions which are
\ACZ-reducible to $\{\alpha\}$.
The fact that \VZalpha\ can define all functions in $\FACZ(\alpha)$ can be proved by induction
on $n$ in Definition \ref{d:reducible}.

For each other class, the ($\Longra$) direction can also be proved by induction on $n$ from Definition \ref{d:reducible}
using Theorem \ref{t:relclass-char}.

The ($\Longla$) direction can be proved by a standard witnessing argument,
i.e., by induction on the length of a free-cut-free
\VCalpha-proof whose end-sequent is
the defining axiom of a function provably total in \VCalpha.
\end{proof}
}

\subsubsection{Classes with unbounded oracle nested depth}
Now we present the theories $\VACkalpha$ (for $k \ge 1$) and $\VNCkalpha$ (for $k \ge 2$).
For the nonrelativized case, the axioms for the theories use the fact
that the problem of evaluating an unbounded fanin (resp. bounded fanin) circuit of depth $(\log n)^k$
is \ACZ-complete for \ACk\ (resp. \NCk).
Unfortunately these nonrelativized problems are not
$\ACZ(\alpha)$-complete for the corresponding relativized problems, 
unlike the situation for the classes with bounded oracle nested depth
considered previously.  However for the oracle versions
of the circuit classes, the evaluation problems become \ACZ-complete
for the corresponding relativized classes, provided (in the case of \NCk) 
the circuit descriptions tell the nested oracle depth of each
oracle gate.
Thus $\VACkalpha$ (or \VNCkalpha) will be axiomatized by \VZ\ together with
an additional axiom that formalizes an oracle computation that solves
the respective complete problem.

First we describe the encoding of the input.
As before, a circuit of width $w$ and depth $d$ will be encoded by $(w,d,E,G)$,
and its input will be denoted by $I$.
Since the order of inputs to an oracle gate is important,
the string variable $E$ that encodes the wires in the circuit is now four-dimensional:
$E(z,u,t,x)$ indicates that gate $(z,u)$ (i.e. the $u$-th gate on layer $z$)
is the $t$-th input to gate $(z+1,x)$.
Also, the type of a gate $(z,x)$ is specified by $(G)^{\tuple{z,x}}$ as before, but now it
can have value in $\wedge$, $\vee$, $\neg$, or $\alpha$
(we no longer consider just monotone circuits).
We use the following formula to ensure that this is a valid encoding;
it says that each gate $(z+1,x)$ has an arity $s$ which is 1 if the gate is a $\neg$-gate.
Moreover for $t < s$ the $t$-th input to a gate is unique.
\begin{multline*}
\Proper(w,d,E,G) \equiv 
\forall z < d \forall x < w \exists! s \le w (s\ge 1 \wedge \Arity(z+1,x,s,E)) \wedge\\
 (G)^{\tuple{z+1,x}} = \mbox{``$\neg$''} \supset \Arity(z+1,x,1,E)
\end{multline*}
where
$\Arity(z+1,x,s,E)$ (which asserts that gate $(z+1,x)$ has arity $s$) 
is the formula:
\begin{equation}\label{e:Arity}
\forall t < s \exists! u < w\ E(z,u,t,x) \wedge\\
\forall t < w (s \le t \supset \neg\exists u < w E(z,u,t,x))
\end{equation}

The formula $\delLocv^{\alpha}(w,d,E,G,I,Q,Y)$ defined below states that
$(Q,Y)$ is an evaluation of the oracle circuit $(w,d,E,G)$ on input $I$.
Here the string $Q^{[z+1,x]}$ encodes the query to the oracle gate
$(z+1,x)$ and bit $Y(z,x)$ is the value of gate $(z,x)$.
(LOCV stand for ``layered oracle circuit value.'')

\begin{definition}
The formula $\delLocv^{\alpha}(w,d,E,G,I,Q,Y)$ is the formula
\begin{multline*}\label{e:Lmcvalpha}
\forall z < d \forall x < w,\ \ \
[Y(0,x) \lra I(x)] \wedge  \\
\big[\forall t < w(Q^{[z+1,x]}(t) \lra (\exists u < w,\, E(z,u,t,x) \wedge Y(z,u)))\big]\ \wedge  \\
\big[Y(z+1, x) \lra
        \big(((G)^{\tuple{z+1,x}} =\mbox{``$\wedge$''} \wedge \forall t,u < w,\, E(z,u,t,x) \supset Y(z,u)) \vee\\
        ((G)^{\tuple{z+1,x}} = \mbox{``$\vee$''} \wedge \exists t < w \exists u < w,\, E(z,u,t,x) \wedge Y(z,u)) \vee\\
        ((G)^{\tuple{z+1,x}} = \mbox{``$\neg$''} \wedge \exists u < w,\, E(z, u, 0, x) \wedge \neg Y(z,u)) \vee\\
	((G)^{\tuple{z+1,x}} = \mbox{``$\alpha$''} \wedge \alpha(Q^{[z+1,x]}))\big)\big]
\end{multline*}
\end{definition}

\begin{definition}[\VACkalpha]
For $k \ge 1$, \VACkalpha\ is the theory over the vocabulary
$\LTwoA(\alpha)$
and is axiomatized by the axioms of \VZ\ and the following axiom:
\begin{equation}\label{e:VACax}
(\Proper(w,d,E,G) \supset \exists Q \exists Y\, \delLmcv^{\alpha}(w,|w|^k,E,G,I,Q,Y))
\end{equation}
\end{definition}

The axiom (\ref{e:VACax}) asserts that \VACk\ circuits can be evaluated.
Since a function in $\FACk(\alpha)$ is computed by an \ACZ-uniform
family of $\ACk(\alpha)$ circuits, and our method of describing
an oracle circuit by the tuple $(w,d,E,G)$ can be taken as a definition,
the axiom is clearly strong enough to show that the functions in
$\FACk(\alpha)$ are $\SigOneB(\alpha)$-definable in $\VACk(\alpha)$.
But in order to show the converse we need
to show that the existential quantifiers $\exists Q \exists Y$ can
be witnessed by functions in $\FACk(\alpha)$, and for this
we need to show the
existence of universal circuits for $\ACk(\alpha)$.  This is done
in the next result.

\begin{proposition}\label{t:ACk-alpha}
For $k \ge 1$, the problem of evaluating the circuit encoded by
$(w,|w|^k, E, G)$ on a given input $I$, assuming
$\Proper(w,|w|^k, E, G)$ is satisfied,
is complete for $\FACk(\alpha)$ under \ACZ-many-one reductions.
\end{proposition}

\begin{proof}
The hardness direction follows by the discussion above:
Every function $F(X)$
in $\FACk(\alpha)$ can be computed by a circuit family in which the
parameters $w,E,G,I$ for each circuit can be computed by \ACZ\
functions of $X$.

Conversely we need to prove membership of the circuit evaluation problem
in  $\FACk(\alpha)$.  We do this for the case $k=1$.  The proof for the
general case is similar.
Thus we need to construct a universal circuit
for oracle circuits of depth $\log n$.
In fact, we will construct a universal circuit
(of depth $\bigO(d)$, size polynomial in $w,d$) for all circuits of depth $d$ and width $w$.
Let $C = (w,d,E,G)$ denote the given circuit.
The idea is to construct a component $K_{z,x}$ for each gate $(z,x)$
in $C$, where $z < d$ and $x < w$:
$K_{z+1,x}$ is an $\ACZ(\alpha)$ circuit that takes inputs from $E$, $G$, $I$ and $K_{z,u}$ for all $u<w$,
so that when each $K_{z,u}$ computes gate $(z,u)$ in $C$,
$K_{z+1,x}$ computes the value of gate $(z+1,x)$.
We will present $K_{z,x}$ as a bounded depth formula.

The circuits $K_{0,x}$ are easy to define:
for all $x < w$, $K_{0,x} \equiv I(x)$.
For $z \ge 0$, $K_{z+1,x}$ is the following disjunction:
\begin{multline*}
(G)^{\tuple{z+1,x}} = \mbox{``$\wedge$''}
	\wedge \bigwedge_{t < w} \bigwedge_{u < w} E(z,u,t,x) \supset K_{z,u} \vee \\
(G)^{\tuple{z+1,x}} = \mbox{``$\vee$''}
	\wedge \bigvee_{t < w} \bigvee_{u < w} E(z,u,t,x) \wedge K_{z,u} \vee \\
(G)^{\tuple{z+1,x}} = \mbox{``$\neg$''} \wedge \bigvee_{u < w} E(z,u,0,x) \wedge \neg K_{z,u} \vee \\
\big[(G)^{\tuple{z+1,x}} = \mbox{``$\alpha$''} \wedge \bigvee_{s < w}\\
	\big(\Arity(z,x,s,E) \wedge
	\alpha(\bigvee_{u<w} (E(z,u,0,x) \wedge K_{z,u}), \ldots,
	\bigvee_{u<w} (E(z,u,s-1,x) \wedge K_{z,u}))\big)\big]
\end{multline*}
Now by arranging $K_{z,x}$ in the same order as $(z,x)$ we obtain an
$\ACOne(\alpha)$ circuit that evaluates $C$.
\end{proof}

\begin{theorem}\label{t:VACk-def}
For $k \ge 1$, the functions in $\FACk(\alpha)$ are precisely the
$\SigOneB(\alpha)$-definable functions of $\VACk(\alpha)$.
\end{theorem}

\begin{proof}
By Proposition \ref{t:ACk-alpha} the problem of witnessing the
quantifiers $\exists Q\exists Y$ in the axiom (\ref{e:VACax}) is
in $\FACk(\alpha)$, and so by a standard witnessing argument every
$\SigOneB(\alpha)$-definable function is in $\FACk(\alpha)$.
The converse follows from the hardness direction of Proposition
\ref{t:ACk-alpha} and the fact that the $\SigOneB(\alpha)$-definable
functions in $\VACk(\alpha)$ are closed under $\ACZ$-reductions, by
the methods used in Chapter IX of \cite{Cook:Nguyen}.
\end{proof}

Finally we consider $\NCkalpha$ classes for $k\ge 2$.
To specify an $\NCkalpha$ circuit, we need to express the condition that
$\wedge$ and $\vee$ gates have fanin 2.
We use the following formula $\Fanin'(w,d,E,G)$ to express this,
see also \eqref{e:Arity}:
\begin{multline*}
\forall z < d \forall x < w,\ \
\big((G)^{\tuple{z,x}} \neq \mbox{``$\alpha$''} \wedge (G)^{\tuple{z,x}} \neq \mbox{``$\neg$''} \big)
\supset Arity(z,x,2,E)
\end{multline*}

Moreover the nested depth of oracle gates in circuit $(w,d,E,G)$ needs
to be bounded separately from the circuit depth $d$.
We use a formula $\OHeight_k(w,d,E,G,D)$ which states that this nested depth is bounded by $|w|^k$.
Here the extra string variable $D$ is to compute the nested depth of oracle gate:
$D$ is viewed as a sequence, where $(D)^{\tuple{z,x}}$ is the oracle depth of gate $(z,x)$.
(Recall that the gates $(0,x)$ are input gates.)
The sequence is computed inductively, starting with the input gates.
An explicit formulation is rather straightforward but tedious,
so we omit the details here.
Note that we can use \ACZ\ number functions
such as $|x|$ and $\max$ (which returns the maximum element in a
bounded sequence),
because they can be eliminated, see \cite[Lemma V.6.7]{Cook:Nguyen}.
\ignore{
\begin{gather*}
\OHeight_k(w,d,E,G,D) \ \equiv \
\forall z \le d \forall x<w\ \\
(D)^{\tuple{0,x}} = 0 \wedge (D)^{\tuple{d,x}} \le (\log w)^k \wedge \\
(G)^{\tuple{z+1,x}} = \mbox{``$\alpha$''} \supset 
(D)^{\tuple{z+1,x}} = 1  + \max\{(D)^{\tuple{z,u}}: u < w \wedge \exists t < w\, E(z,u,t,x)\}\ \wedge \\
(G)^{\tuple{z+1,x}} = \mbox{``$\neg$''} \supset 
(D)^{\tuple{z+1,x}} = (D)^{\tuple{z,u}}: u < w \wedge \exists t < w\, E(z,u,t,x)\}
(G)^{\tuple{z+1,x}} \neq \mbox{``$\alpha$''} \supset 
(D)^{\tuple{z+1,x}} = \max\{(D)^{\tuple{z,u}}: u < w \wedge \exists t < w\, E(z,u,t,x)\}
\end{gather*}
}

\begin{definition}[\VNCkalpha]
For $k\ge 2$, $\VNCkalpha$ is the theory over $\LTwoA(\alpha)$ and is axiomatized by \VZ\
and the axiom
\begin{multline}\label{e:axVNCK}
[\Proper(w,d,E) \wedge \Fanin'(w,|w|^k, E, G) \wedge  \\
\OHeight_{k-1}(w,d, E, G, D)]\ \supset\
 \exists Q \exists Y\ \delLocv^{\alpha}(w,|w|^k,E,G,I,Q,Y)
\end{multline}
\end{definition}

\begin{proposition}\label{t:NCk-alpha}
For $k \ge 2$,  the problem of witnessing the quantifiers
$\exists Q\exists Y$ in the axiom for \VNCkalpha\
is complete for $\NCkalpha$ under \ACZ-many-one reductions.
\end{proposition}


First we exhibit a problem complete under \ACZ-reductions for \NCOnealpha.
Informally, this is the problem of evaluating a relativized sentence which is given using
the {\em extended connection language} \cite{Ruzzo:81:jcss}.
More precisely, we consider encoding relativized sentences by tuples $(a,G,I,J)$ in the following way.
The sentence is viewed as a balanced binary tree as in the axiom $\MFV$ \eqref{e:MFV},
but now each leaf $Y(x+a)$ can be an input bit (from $I$) or (the negation of) an $\alpha$-gate
that takes its input from $J$.
In other words, the underlying circuit for the sentence has exactly one layer of oracle gates
which take input directly from the input constants.
More precisely, let 
\begin{multline*}
\delta(a,G,I,J,Y) \equiv \forall x < a,\
	G(x+a) = \text{``$\alpha$''} \supset (Y(x+a) \lra \alpha(J^{[x]})) \ \wedge \\
	G(x+a) = \text{``$\neg\alpha$''} \supset (Y(x+a) \lra \neg\alpha(J^{[x]})) \ \wedge\\
	G(x+a) = \text{``const''} \supset (Y(x+a) \lra I(x)) \ \wedge\\
        0 < x\ \supset\ Y(x) \lra \left(\big(G(x) \wedge Y(2x) \wedge Y(2x+1)\big) \vee \right.\\
        \left. \big(\neg G(x) \wedge (Y(2x) \vee Y(2x+1))\big)\right)
\end{multline*}
In the next result we emphasize that the \ACZ-reductions referred to
are the `Turing' reductions given in Definition \ref{d:reducible}.
\begin{lemma}\label{t:NCOne-alpha}
The relation given by the formula $\exists Y(\delta(a,G,I,J,Y) \wedge Y(1))$ is
\ACZ-complete for \NCOnealpha.
\end{lemma}
\begin{proof}
For the hardness direction we note that the circuits solving a problem
in \NCOnealpha\ have oracle nested depth bounded by some constant $d$.
Hence such a circuit can be simulated by $d$ circuits of the form
described above, forming $d$ layers.  The layers can be evaluated by
$d$ successive queries to
the relation in the lemma, where in each layer except the first, the
constant inputs $I$ and the oracle inputs $J$ are determined by the gate
values in the previous layer.

For membership in \NCOnealpha, observe that we can evaluate the first layer of the circuit
by an \ACZalpha\ circuit (see also the proof of Proposition \ref{t:ACk-alpha}).
Once this has been done, the remaining task is to evaluate a nonrelativized,
balanced, monotone boolean sentence, which can be done by an \NCOne\ circuit.
\end{proof}

\begin{proof}[Proof outline of Proposition \ref{t:NCk-alpha}]
The hardness direction is proved as for Proposition \ref{t:ACk-alpha}:
We assume that by definition an \NCkalpha-circuit must satisfy the
hypotheses of the axiom (\ref{e:axVNCK}).

Now we argue that the problem actually belongs to $\NCkalpha$.
Consider the case $k = 2$; other cases are similar.
First, the given problem reduces to the following restriction of it, called $P$,
where the layers in the given circuit are grouped together to form $|w|$
many blocks $B_1,B_2,\ldots, B_{|w|}$,
where each block $B_i$ has exactly $|w|$ layers and $w$ outputs.
Furthermore, each block is an \NCOnealpha\ circuit (with multiple outputs)
such that all $\alpha$-gates appear in the first layer.
Moreover, these \NCOnealpha\ circuits are presented using the extended connection language.
The reduction can be done by uniform circuits of polynomial size, $\log\log n$ depth and unbounded fanin,
where $n$ is the length of the input to our original problem.

It remains to show that the new problem $P$ is solvable by a uniform family of $\NC^2(\alpha)$ circuits.
Note that the input now can be viewed as the sequence 
$$B_1, B_2, \ldots, B_{|w|}$$
where each $B_i$ consists of $w$ single-output \NCOnealpha\ circuits
$$B_{i,1}, B_{i,2}, \ldots, B_{i,w}$$
Here each $B_{i,j}$ is an \NCOnealpha\ circuit where all $\alpha$-gates are on the first layer.
Lemma \ref{t:NCOne-alpha} above shows that each $B_{i,j}$ can be evaluated by an $\NCOnealpha$ circuit $C_{i,j}$.
As a result, the circuits for solving $P$ are obtained by arranging $C_{i,j}$ appropriately.
\end{proof}

The next theorem is proved in the same way as Theorem \ref{t:VACk-def},
using Proposition \ref{t:NCk-alpha}.

\begin{theorem}\label{t:NCk-definable}
For $k \ge 2$, the functions in $\FNCk(\alpha)$ are precisely the
$\SigOneB(\alpha)$-definable functions of $\VNCk(\alpha)$.
\end{theorem}

Now we can apply the separations of the relativized classes obtained in
Section \ref{s:sep} to prove separations of the corresponding
theories.

\begin{corollary}
$\VLalpha \subsetneq \VACOne(\alpha)$,
and for $k\ge 1$:
$$\VACkalpha \subseteq \VNC^{k+1}(\alpha) \subsetneq \VAC^{k+1}(\alpha)$$
\end{corollary}

\begin{proof}
The first inclusion follows from Proposition \ref{t:VTC-alpha-alt},
and the fact that the axiom $\PATH(\alpha)$ is implied by the
axiom for $\VACOne(\alpha)$.
The remaining inclusions are easy to check, so it suffices to show the 
strictness of the strict inclusions.
By Theorems \ref{t:def-VNL}, \ref{t:VACk-def}, and \ref{t:NCk-definable}
we know that the $\SigOneB(\alpha)$-definable functions in each
theory are those in the corresponding complexity class.  
By Corollary \ref{c:singlestrict} we know that the inclusions of
the corresponding complexity classes are strict, where indicated
in the statement of the corollary.
\end{proof}

\section{Conclusion}
\label{s:conclusion}
The the relativized class $\ACk(\alpha)$, $k\ge 0$, has an obvious
definition:  treat an oracle gate $\alpha(x_1,\ldots,x_n)$ in
the same way as $\wedge$ and $\vee$ gates.  However definitions of the
relativized versions of the classes \NCk, \L, and \NL\
are not so obvious.  Here we give new definitions for these classes
that preserve many of the properties of the unrelativized classes,
namely class inclusions, Savitch's Theorem, and the Immerman-\Sze\ Theorem.
However there is a weakness in our definition of $\csNLalpha$
(relatived \NL), namely the corresponding function class may not
be closed under composition (all other function classes are so closed).
A possible way out is to define relativized \NL\ to be
$\ACZ(\STCONN, \alpha)$ (see Theorem \ref{t:order-b} part (ii) and
its proof).
This class has nice closure properties and satisfies the expected
inclusions with other relativized classes.  It also has a natural
associated relativized theory, namely $\VNLprimealpha$ (see
Definition \ref{d:relTheories}).  But we do not know
how to define $\ACZ(\STCONN, \alpha)$ in terms of nondeterministic
log space oracle Turing machines. 
We leave this conundrum as an open problem.

We note that the first author has carried out in \cite{Aehlig:Hab}
a detailed study of propositional versions of our relativized theories.


\bibliographystyle{alpha}
\bibliography{ntp}



\end{document}